\documentclass[conference,10pt]{IEEEtran}
\IEEEoverridecommandlockouts
\usepackage[english]{babel}
\usepackage{cite}
\usepackage{tikz,graphicx,subfigure,epstopdf}
\usepackage{lipsum,setspace,multirow,array}
\usepackage[small,sc]{caption}
\usepackage{amsmath,amsfonts,amssymb,amscd,amsthm,amstext}
\usepackage{mleftright}
\newcommand{\pren}[1]{\mleft(#1\mright)}

\theoremstyle{plain}
\newtheorem{thm}{Theorem}
\newtheorem{prop}{Proposition}
\newtheorem{cor}{Corollary}
\newtheorem{lem}{Lemma}

\theoremstyle{definition}
\newtheorem{defn}{Definition}
\newtheorem*{defn*}{Definition}

\theoremstyle{remark}
\newtheorem*{remark*}{Remark}

\DeclareMathOperator*{\one}{\mathbf{1}}

\DeclareMathOperator*{\mcalr}{\mathcal{R}}
\renewcommand{\d}[1]{\ensuremath{\operatorname{d}\!{#1}}}

\newcommand{\R}{\mathbb{R}}

\newcommand{\Cc}{\mathcal{C}}
\newcommand{\Pp}{\mathcal{P}}

\newcommand{\E}{\mathbb{E}}
\newcommand{\Pb}{\mathbb{P}}

\newcommand{\bfx}{\mathbf{x}}

\newcommand{\tail}{\bar{F}}

\newcommand{\tx}{\textrm{tx}}
\newcommand{\rx}{\textrm{rx}}
\newcommand{\sinr}{\mathrm{SINR}}

\newcommand{\fref}[1]{Figure~\ref{#1}}
\newcommand{\tref}[1]{Table~\ref{#1}}
\newcommand{\eref}[1]{(\ref{#1})}

\newcommand{\aref}[1]{Appendix~\ref{#1}}
\newcommand{\lemref}[1]{Lemma~\ref{#1}}
\newcommand{\thmref}[1]{Theorem~\ref{#1}}
\newcommand{\propref}[1]{Proposition~\ref{#1}}
\newcommand{\corref}[1]{Corollary~\ref{#1}}

\begin{document}

\title{Coverage and Capacity Scaling Laws in Downlink Ultra-Dense Cellular Networks}
\author{\IEEEauthorblockN{Van Minh Nguyen and Marios Kountouris}%
	\IEEEauthorblockA{Mathematical and Algorithmic Sciences Lab, France Research Center, Huawei Technologies Co. Ltd. \\
		\{vanminh.nguyen, marios.kountouris\}@huawei.com}}
\maketitle

\begin{abstract}
    Driven by new types of wireless devices and the proliferation of bandwidth-intensive applications, data traffic and the corresponding network load are increasing dramatically. Network densification has been recognized as a promising and efficient way to provide higher network capacity and enhanced coverage. Most prior work on performance analysis of ultra-dense networks (UDNs) has focused on random spatial deployment with idealized singular path loss models and Rayleigh fading. In this paper, we consider a more precise and general model, which incorporates multi-slope path loss and general fading distributions. We derive the tail behavior and scaling laws for the coverage probability and the capacity considering strongest base station association in a Poisson field network. Our analytical results identify the regimes in which the signal-to-interference-plus-noise ratio (SINR) either asymptotically grows, saturates, or decreases with increasing network density. We establish general results on when UDNs lead to worse or even zero SINR coverage and capacity, and we provide crisp insights on the fundamental limits of wireless network densification.
\end{abstract}

\section{Introduction}\label{s:Introduction}
With exponential increase in data traffic driven by a new generation of wireless devices, data is expected to overwhelm cellular network capacity in the near future. Heterogeneous cellular network (HetNet) deployment is a promising and effective way to provide high cellular network capacity by overlaying conventional macrocell cellular architecture with heterogeneous architectural features, such as small cellular access points (picocells and femtocells), low-power fixed relays, and distributed antennas. Ultra-dense networks (UDNs) are expected to achieve higher data rates and enhanced coverage by exploiting spatial frequency reuse, while retaining at the same time seamless connectivity and mobility. Inspired by the attractive features and potential advantages of UDNs, their development and deployment has been gaining momentum in both wireless industry and research community during the last few years. It has also attracted the attention of standardization bodies, e.g. 3GPP LTE-Advanced. 

Existing cellular network analyses have focused on stochastic geometry based models, in which base stations (BSs) are located according to a Poisson point process (PPP). Most prior results have considered a standard singular power-law path loss model and Rayleigh fading, as a means to provide tractable analysis of coverage probability and other key performance metrics in downlink cellular networks. For instance, recent results using the aforementioned models with closest BS \cite{Andrews2011} and strongest BS association \cite{ Dhillon2012} show that the coverage probability does not depend on the network density when thermal noise is negligible. However, the singular single-slope path loss model used therein leads to unrealistic results in certain scenarios and fails to accurately capture the dependence of path loss exponent on the link distance. In \cite{Zhang2015}, the authors study the impact of dual slope path loss on the performance of downlink UDNs under closest BS association and show that both coverage and capacity performance strongly depends on the network density. More precisely, it is shown that the network coverage in terms of signal-to-interference-plus-noise ratio (SINR) is maximized at some finite density and there exists a phase transition in the near-field path loss exponent with ultra densification (i.e. network density goes to infinity). In \cite{Chen2012}, the authors consider strongest BS association with bounded path loss and lognormal fading and show that the coverage attains a maximum point before going down when the network densifies. Based on system-level simulations, \cite{LopezPerez2015} shows that there is a fundamental limit of 1 cell/user in UDNs, although such deployments are neither cost nor energy efficient. Despite shedding light on the performance limits of UDNs and the optimal network density, previous work has mainly focused on the effect of path loss singularity \cite{Inaltekin2009, Haenggi2009, Nguyen2011}  or boundedness \cite{Zhang2015}. The effect of fast fading has been only investigated in some particular contexts, e.g. \cite{Haenggi2009}, \cite[Chap.~8]{Nguyen2011}. Furthermore, emerging utilization of advanced communication and signal processing techniques is expected to enhance the channel gain, which in some cases may have a regularly varying tail \cite{Rajan2015} or a diffuse power component \cite{Durgin2002}.

In this work, we analyze the SINR, coverage probability and capacity of downlink UDNs under multi-slope path loss and general fading distributions, considering strongest BS association in PPP networks. In particular, using general fading distributions, we provide the ability to capture the effect of either fast fading, shadowing, or composite fast fading-shadowing. We study the tail behavior of the received signal power, as well as the SINR, coverage and capacity scaling in the ultra-dense regime. Our results 
provide useful insights on the fundamental limits of network densification with the following main conclusions:
\begin{itemize}
	\item Under the Poisson field assumption, the most affecting component of the path loss is its near-field exponent $\beta_0$. Bounded path loss (obtained for $\beta_0 = 0$) is just a special case of $\beta_0 < d$ with $d$ being the network dimension. 
	\item The effect of fading on the performance scaling is as significant as path loss, and regularly varying fading distributions have the same effect as path loss singularity.
	\item In more conventional cases when the near-field path loss exponent is greater than the free-space dimension (i.e. $\beta_0 > d$), or when fading is heavy tailed (i.e. $\tail_m \in \mcalr_{-\alpha}$, $\alpha \in (0,1)$), the coverage and capacity saturate at a limiting bound when the network density increases.
	\item In more realistic cases with $\beta_0 < d$ (in particular $\beta_0 = 0$) and fading being less heavy tailed or even truncated, both coverage and capacity exhibit an `inverse U' behavior w.r.t. network density, i.e. both are maximized at finite density then vanish when the network further densifies.
	\item Finally, all standard fading models, such as Rayleigh, lognormal, Gamma, and their composite forms, belong to the same class of fading distributions, which leads to coverage and capacity maximization at a certain finite network density.
\end{itemize}  


\section{System Model}\label{s:Model} 

\subsection{Network Model} 
Consider a typical downlink user located at the origin and that the network is composed of cell sites located at positions $\{\bfx_i, i = 0,1,\ldots\}$. For convenience, cell sites are referred to as nodes, whereas the \emph{typical user} is simply referred to as \emph{user}. Unless otherwise stated, $\{\bfx_i\}$ are assumed to be random variables independently distributed on the \emph{network domain} according to a homogeneous Poisson point process (PPP) of intensity $\lambda$, denoted by $\Phi$. In prior work, the entire $d$-dimensional Euclidean space $\R^d$, where $d = 2$ is usually assumed as network domain. Since network domain is in reality limited, and far-away nodes are less relevant to the typical user due to path loss, we assume that the distance from the user to any node is upper bounded by some constant $0 < R_{\infty} < \infty$, which can be arbitrarily large. Each node transmits with some power that is independent to the others but is not necessarily constant. 

\subsection{Propagation Model} 
Let $l: \R_+ \to \R_+$ represent the path loss function. The receive power $P_{\rx}$ is related to the transmit power $P_{\tx}$ by $P_{\rx} = P_{\tx}/l(r)$ with $r$ being the transmitter-receiver distance. Physics laws require that $l(r) \leq 1, \forall r$. However, in the literature, $l(\cdot)$ has been usually assumed to admit a power-law model, i.e. $l(r) \sim r^{\beta}$ where $\beta$ is the path loss exponent satisfying $\beta \geq d$. This far-field propagation model has been widely used mainly due to its tractability. However, for short ranges, especially when $r \to 0$, this model is no longer relevant and becomes singular at the origin. In the context of network densification where the inter-site distance becomes smaller, the above singular model may be unsuitable. Moreover, the dependence of the path loss exponent on the distance in emerging millimeter wave (mmWave) communications \cite{Deng2015,Zhang2015} advocates the use of a more generic path loss function. In this work, the path loss is modeled as follows
\begin{equation}\label{eq:PL}
    l(r) = \sum_{k=0}^{K-1} A_k r^{\beta_k} \one(R_k \leq r < R_{k+1}),
\end{equation}
where $\one(\cdot)$ is the indicator function, $K \geq 1$ is a given constant characterizing the number of path loss slopes, $R_k$ are constants satisfying
\begin{equation}\label{eq:plRange}
    0 = R_0 < R_1 < \ldots < R_{K-1} < R_K = R_{\infty},
\end{equation}
$\beta_k$ denotes the path loss exponent satisfying
\begin{subequations}\label{eq:plExp}
\begin{align}
    \beta_0 & \geq 0, \label{eq:plExpA}\\
    \beta_k & \geq d-1, \text{ for } k = 1,\ldots,K-1, \label{eq:plExpB}\\
    \beta_k & < \beta_{k+1}, \text{ for } k = 0,\ldots,K-2, \label{eq:plExpC}
\end{align}
\end{subequations}
and $A_k$ are constants to maintain continuity of $l(\cdot)$, i.e.
\begin{equation}\label{eq:plScale}
    A_k > 0, \text{ and } A_k R_{k+1}^{\beta_k} = A_{k+1} R_{k+1}^{\beta_{k+1}},
\end{equation}
for $k = 0,\ldots,K-2$. For notational simplicity, we also use the following notation
\begin{equation}
    \alpha_k = d/\beta_k, \quad \text{for } k = 0,\ldots,K-1.
\end{equation}
The above model (cf. \eref{eq:PL}) captures that the path loss exponent varies with the distance while remaining unchanged within a certain range. In principle, free-space propagation in $\R^3$ has path loss exponent equal to 2 (i.e. $\beta = d-1$), whereas in realistic scenarios, path loss models often include antenna imperfections and empirical models usually result in the general condition \eref{eq:plExpB} for far-field propagation. Condition \eref{eq:plExpC} models the physical property that the path loss increases faster as the distance increases. Notice, however, that this condition is not important in the subsequent analytical development. Finally, condition \eref{eq:plExpA} is related to the near field (i.e. it is applied to the distance range $[0, R_1]$). 

The multi-slope path loss function, as defined above, has the following widely used special cases:
\begin{itemize}
    \item $K=1$, $\beta_0 \geq d$: $l(r) = A_0 r^{\beta_0}$, which is the standard unbounded path loss,
    \item $K=2$, $\beta_0 = 0$: $l(r) = \max(A_0, A_1r^{\beta_1})$, which is the bounded path loss recommended by the 3GPP standard, in which $A_0$ is referred to as the minimum coupling loss.
\end{itemize}
Due to the particular importance of path loss boundedness so as to have a realistic model, we have the following definitions:
\begin{defn}
    A path loss function $l: \R_+ \to \R_+$ is said \emph{bounded} if and only if $1/l(r) < \infty, \forall r \in \R_+$, and \emph{unbounded} otherwise. Furthermore, the path loss function $l(\cdot)$ is said \emph{physical} if and only if  $1/l(r) \leq 1, \forall r \in \R_+$.
\end{defn}
It is clear that the path loss function \eref{eq:PL} is bounded if and only if (iff) $\beta_0 = 0$, and is physical iff $\beta_0 = 0$ and $A_0 \geq 1$.

Besides path loss, shadowing and fast fading are also sources of wireless link variations, which are commonly referred to as \emph{fading} in the sequel. Let $m_i$ be a variable containing transmit power, fading, and any gains or attenuation other than path loss from $i$-th node to the user. Given node location $\{\bfx_i\}$, the variables $\{m_i\}$ are assumed not identical to zero and independently distributed according to some distribution $F_m$. 
To this end, the signal power, say $P_i$, that the user receives from $i$-th node is expressed as $P_i = m_i / l(||\bfx_i||)$ where $||\cdot||$ is the Euclidean distance.

\subsection{Performance Metrics}
The quality of the signal received from $i$-th node is expressed in terms of its SINR as
\begin{equation*}
	\sinr_i = P_i/(I_i + W),
\end{equation*}
where $I_i = \sum_{j \neq i} P_j$ is the aggregate interference with respect to $i$-th node's signal, and thermal noise at the user's receive antenna is assumed Gaussian with average power $W$. In addition, we define $I = \sum_{j} P_j$ to be the total interference. The main metrics used in this paper are the coverage probability and the average rate that the user experiences from its \emph{serving cell}. Let $Y$ denote the SINR that the user receives from its serving cell. The \emph{SINR coverage probability}, denoted by $\Pp_y$, is defined as the probability that $Y$ is larger than a given target $y$, and the \emph{capacity}, denoted by $\Cc$, is defined as the Shannon rate (in nats/s/Hz) assuming Gaussian codebooks, i.e.
\begin{equation}\label{eq:PerfDefn}
	\Pp_y = \Pb(Y \geq y), \text{ and } \Cc = \E(\log(1+Y)).
\end{equation}

\subsection{User Association} 
The above performance metrics are defined with respect to the user's serving cell, which in turn depends on the underlying user association scheme. \emph{Strongest cell association} is used in this work, i.e. the user is connected to the cell that provides the best/strongest signal quality (in practice, time averaging of the signal is usually performed to avoid frequent handover due to fast fading). In a longer version of this work, we also investigate the case of nearest cell association and show the effect of user association on performance scaling. Under strongest cell association, the SINR of the user is given by $Y = \max_i \sinr_i$ and can be expressed as \cite{Nguyen2010}
\begin{equation}\label{eq:YMI}
    Y = M/(I + W - M), \quad \text{with } M = \max_{i}P_i.
\end{equation}

\subsection{Notation} 
Quantities whose dependence on the network density $\lambda$ is important, are denoted as $\cdot(\lambda)$, e.g. $Y(\lambda)$, $I(\lambda)$, and $M(\lambda)$.
We also denote by $r$, $m$, and $P$, the distance, associated fading, and received power from a random node, respectively. Let $F_P$ be the distribution of $P$, and $\tail_P = 1 - F_P$. In addition, for real functions $f$ and $g$, we say $f = O(g)$ if $\lim_{x \to \infty}(f(x)/g(x)) = c$ for $c \in (0, \infty)$, $f \sim g$ if $\lim_{x \to \infty}(f(x)/g(x)) = 1$, and $f = o(g)$ if $\lim_{x \to \infty}(f(x)/g(x)) = 0$. We also use notation $\overset{d}{\to}$, $\overset{p}{\to}$, $\overset{a.s}{\to}$ to denote the convergence in distribution, convergence in probability, and almost sure (a.s) convergence, respectively. Finally, for two random variables $X_1$ and $X_2$ defined on the same probability space, we say that $X_1$ is \emph{statistically} greater than $X_2$, denoted by $X_1 \overset{st}{>} X_2$, if $\Pb(X_1 \geq x) > \Pb(X_2 \geq x)$ $\forall x$. Similar definition is for $\overset{st}{\geq}$, $\overset{st}{<}$, and $\overset{st}{\leq}$.

\begin{defn}[Regular variation \cite{Embrechts1997}] 
	A positive, Lebesgue measurable function $h$ on $(0,\infty)$ is called \emph{regularly varying} with index $\alpha \in \R$ at $\infty$ if $\lim_{x \to \infty} \frac{h(tx)}{h(x)} = t^{\alpha}$ for $t > 0$. In particular, $h$ is called \emph{slowly varying} (\emph{rapidly varying}, resp.) (at $\infty$) if $\alpha = 0$ (if $\alpha = -\infty$, resp.). We denote by $\mcalr_{\alpha}$ the class of regularly varying functions with index $\alpha$.
\end{defn}
Note that if $h$ is a regularly varying function with index $\alpha$ at $\infty$, it can be represented as $h(x) = x^{\alpha} L(x)$ as $x \to \infty$ for some slowly varying function $L \in \mcalr_0$.

\begin{defn}[Tail-equivalence]
    Two distributions $F$ and $H$ are called \emph{tail-equivalent} if they have the same right endpoint, say $x_{\infty}$, and $\lim_{x \uparrow x_{\infty}} \bar{F}(x)/\bar{H}(x) = c$ for $0 < c < \infty$.
\end{defn}

\section{Tail Behavior of Received Signal Power}\label{s:TailBeh}
The network performance mainly depends on the received SINR and hence is a function of $Y$, which in turn depends on $M$ and $I$ (cf. \eref{eq:YMI}). The behavior of the maximum $M$ and the sum $I$ is totally determined by that of the received power $P_i$. Therefore, 
we first study the signal power $P_i$, in particular its tail behavior using tools from extreme value theory. 

\begin{prop}\label{prop:DistanceCDF}
    If  $R_{\infty} < \infty$, then the distance from the user to a random node admits a non-degenerate distribution given as $G(r) = \pren{r/R_{\infty}}^d$ for $r \in [0,R_{\infty}]$,
    and $G(r) = 1$ for $r \geq R_{\infty}$.
\end{prop}
\begin{proof}
    Under the assumption that BSs are distributed according to a homogeneous PPP, the nodes of a realization $\phi$ of $\Phi$ are uniformly distributed. 
    Thus, given $\phi$ and the assumption that $R_{\infty} < \infty$, the distance distribution to a random node of $\phi$, say $G(\cdot; \phi)$, is $G(r;\phi) = \pren{r/R_{\infty}}^d$ for $0 \leq r \leq R_{\infty}$, and $G(r;\phi) = 1$ for $r \geq R_{\infty}$. Then, taking $G(r) = \E_{\phi}(G(r;\phi))$, the result follows.
\end{proof}
\begin{remark*}
	First, note that the distribution $G$ is different from the usual \emph{void probability}, which is the distance to the closest node (nearest neighbor). Second, we can see that for unbounded network domains ($R_{\infty} = \infty$), the distance to a random node does not have a non-degenerate distribution. This is because under the PPP assumption, nodes at equal distance increases with the circumference, which tends to infinity when the outer distance tends to infinity, leading to an absorption of nodes. Therefore, using a limited network domain is not only more realistic, but also useful to have a normally behaving distribution of the distance.
\end{remark*}

\begin{prop}\label{prop:tailP}
    Denote $a_{K} = A_{K-1} R_{\infty}^{\beta_{K-1}}$, and for $k = 0,\ldots, K-1$ denote $a_k = A_k R_k^{\beta_k}$, and
    \begin{equation}\label{eq:Jk}
         J_k(t) = \frac{\E\pren{m^{\alpha_k} \one\pren{a_k t \leq m < a_{k+1} t}}}{A_k^{\alpha_k} R_{\infty}^d} t^{-\alpha_k}.
    \end{equation}
    Then,
    \begin{equation}\label{eq:tailP}
        \tail_P(t) = \tail_m\pren{a_K t} + \sum_{k=k_0}^{K-1}J_k(t),
    \end{equation}
    where $k_0=0$ for $\beta_0 > 0$, and $k_0=1$ for $\beta_0 = 0$.
\end{prop}
\begin{proof} See \aref{appx:ProofTail}. \end{proof}

Based on \propref{prop:tailP}, we derive the following main result:
\begin{thm}\label{thm:tailequiv} The tail distribution (CCDF) of the received signal power depends on the tail distribution of the fading and the path loss function as follows:
    \begin{itemize}
     \item If $\tail_m \in \mcalr_{-\alpha}$ with $\alpha \in [0,\infty]$, then $\tail_P \in \mcalr_{-\rho}$ where $\rho = \min(\alpha_0, \alpha)$ with the convention that $\alpha_0 = +\infty$ for $\beta_0 = 0$, and $\min(\infty,\infty) = \infty$.
     \item If $\tail_m = o(\bar{H})$ with $\bar{H} \in \mcalr_{-\infty}$, then $\tail_P(t)$ and $\tail_m(A_0 t)$ are tail-equivalent for $\beta_0 =0$, $\tail_P \in \mcalr_{-\alpha_0}$ for $\beta_0 > 0$.
    \end{itemize}
\end{thm}
\begin{proof} See \aref{appx:ProofTailEquiv}. \end{proof}

\thmref{thm:tailequiv} shows that the tail behavior of the wireless link depends not only on whether path loss is bounded or not, but also on the tail behavior of the fading. More precisely, a key implication of \thmref{thm:tailequiv} is that path loss and fading have interchangeable effects on the tail behavior of the wireless link. This can also be shown using Breiman's Theorem \cite{Bre65} and results from large deviation and product distributions. In particular, when the fading distribution is a regularly varying function, the wireless link is also regularly varying regardless of the path loss function's boundedness. For lighter-tailed fading, the regular variation property of the wireless link is solely imposed by the path loss singularity.

More importantly, \thmref{thm:tailequiv} is a general result and covers all tail behaviors for the fading: case (1) covers the heaviest tails (i.e. $\mcalr_{-\alpha}$ with $0 \leq \alpha < \infty$, e.g. Pareto distributions), as well as the moderately heavy tails (i.e. the class $\mcalr_{-\infty}$, e.g. exponential, normal, lognormal, Gamma distributions). Case (2) covers all remaining tails (e.g. truncated distributions). 
Therefore, for any realistic statistical model and distribution of fast fading and shadowing, \thmref{thm:tailequiv} enables us to characterize the tail behavior of the wireless link, which is essential to understand the behavior of the interference, the maximum received power, and their asymptotic relationship. In wireless communications, the signal distribution $\tail_m$ often involves lognormal or Gamma shadowing and Rayleigh fading, which all belong to the class $\mcalr_{-\infty}$, and the path loss is bounded, thus $\tail_P \in \mcalr_{-\infty}$. As a result, it can be shown that in most relevant cases in wireless UDNs, $\tail_P$ belongs to the maximum domain of attraction of a Gumbel distribution \cite{Nguyen2010a}.

Finally, the above result generalizes prior results: \cite{Haenggi2009} showed that the interference is tail-equivalent with the fading if path loss is bounded and if $\E(m) < \infty$, which is not applicable for $\tail_m \in \mcalr_{-\alpha}$, $\alpha \in [0,1]$. \cite[Chap.~8]{Nguyen2011} showed that under lognormal fading, $\tail_P$ is regularly varying for unbounded path loss, and behaves like a lognormal tail for bounded path loss.


The following result, which is a direct consequence of \thmref{thm:tailequiv}, can be provided in order to better understand the signal power scaling and the interplay between path loss function boundedness and fading.
\begin{cor}\label{cor:classifyTail}
    \begin{itemize}
        \item $\tail_P \in \mcalr_0$ if and only if $\tail_m \in \mcalr_0$,
        \item $\tail_P \in \mcalr_{-\alpha}$ with $\alpha \in (0,1)$ if $\beta_0 > d$ or $\tail_m \in \mcalr_{-\alpha}$,
        \item $\tail_P \in \mcalr_{-\alpha}$ with $\alpha > 1$ if $0 < \beta_0 < d$ and $\tail_m \in \mcalr_{-\rho}$ with $\rho > 1$ or $\tail_m = o(\bar{H})$ with $\bar{H} \in \mcalr_{-\infty}$.
        \item $\tail_P = o(\bar{H})$ with $\bar{H} \in \mcalr_{-\infty}$ if $\beta_0 = 0$ and $\tail_m = o(\bar{H})$.
    \end{itemize}
\end{cor}

\section{Scaling Laws}\label{s:PerfLim}
In this section, we provide the main results of this paper, namely the SINR, coverage probability and capacity scaling, when the network density is asymptotically large. Using results from Section \ref{s:TailBeh}, we investigate and provide insights on the performance limits of network densification.

First, we start by showing that in sparse networks, the signal quality improves for increased node density.
\begin{lem}\label{lem:YatLdaZero}
    Let $0 \leq \lambda_1 < \lambda_2$. If $W > 0$, then $Y(\lambda_2) \overset{st}{>} Y(\lambda_1)$ as $\lambda_2 \to 0^+$.
\end{lem}
\begin{proof}
    As $\lambda \to 0^+$, we have $(I(\lambda)-M(\lambda)) = o(W)$ almost surely. Thus, for $y > 0$,
	\begin{align*}	    
	    \lim_{\lambda_2 \downarrow 0}\Pb\pren{Y(\lambda_2) \geq y} & = \lim_{\lambda_2 \downarrow 0}\Pb\pren{M(\lambda_2) \geq y W} \\
	    & \overset{(a)}{>} \lim_{\lambda_1 < \lambda_2 \downarrow 0}\Pb\pren{M(\lambda_1) \geq y W} \\
	    & = \lim_{\lambda_1 < \lambda_2 \downarrow 0}\Pb\pren{Y(\lambda_1) \geq y},
	\end{align*} 
    where note that $(a)$ is intuitively evident, but a formal proof can be easily obtained using for example \cite[Prop.~2.4.2]{Baccelli2009}.
\end{proof}

\subsection{SINR Scaling}
We provide here the scaling of the received SINR under strongest BS association in the asymptotically large node density regime. 

\begin{thm}\label{thm:YatLdaInfty} Under the multi-slope path loss model and general fading, as $\lambda \to \infty$, the received SINR behaves as
\begin{enumerate}
    \item $Y \overset{p}{\to} \infty$ if $\tail_P \in \mcalr_0$.
    \item $Y \overset{d}{\to} D$ if $\tail_P \in \mcalr_{-\alpha}$ with $0 < \alpha < 1$, where $D$ has a non-degenerate distribution.
    \item $Y \overset{a.s}{\to} 0$ if $\tail_P \in \mcalr_{-\alpha}$ with $\alpha  > 1$ or $\tail_P = o(\bar{H})$ with $\bar{H} \in \mcalr_{-\infty}$.
\end{enumerate}
\end{thm}
\begin{proof}
    For $y \geq 0$,
    \begin{align*}
	    \Pb(Y > y) = \Pb\pren{\frac{I - M + W}{M} < \frac{1}{y}} \to \Pb\pren{\frac{I}{M} - 1 < \frac{1}{y}},
    \end{align*}
    as $\lambda \to \infty$ since $W$ is finite and $m_i$ is not identical to 0.

    If $\tail_P \in \mcalr_0$, $I/M \overset{p}{\to} 1$ due to \cite{Maller1984}. Thus, $\forall y$
    \[\Pb(Y > y) = \Pb\pren{(I/M) - 1 < y^{-1}} = 1, \quad \text{as }\lambda \to \infty.\]

    If $\tail_P \in \mcalr_{-\alpha}$ with $0 < \alpha < 1$, $M/I \overset{d}{\to} R$ as $\lambda \to \infty$ where $R$ has a non-degenerate distribution \cite{Bingham1981}. As a result,
    \[\Pb(Y > y) = \Pb\pren{(I/M) < 1 + y^{-1}} \to D, \quad \text{as }\lambda \to \infty,\]
    where $D$ is a non-degenerate distribution.

    If $\tail_P \in \mcalr_{-\alpha}$ with $\alpha > 1$ or $\tail_P = o(\bar{H})$ with $\bar{H} \in \mcalr_{-\infty}$, we have $\E(P) < \infty$. Hence, $M/I \overset{a.s}{\to} 0$ due to \cite{OBrien1980}. Moreover, $\forall y \in (0,\infty)$
    \begin{equation*}
    \Pb\pren{Y > y} = \Pb\pren{\frac{M}{I+W} > \frac{y}{1+y}} \leq \Pb\pren{\frac{M}{I} > \frac{y}{1+y}}.
    \end{equation*}
    Thus, $M/I \overset{a.s}{\to} 0$ leads to $Y \overset{a.s}{\to} 0$ as $\lambda \to \infty$.
\end{proof}

Let us have a closer look at \thmref{thm:YatLdaInfty} and on its implications in the interplay between multi-slope path loss and fading. According to \corref{cor:classifyTail}, $\tail_P \in \mcalr_0$ due to the fact that $\tail_m \in \mcalr_0$. 
Recall that $m$ contains the transmit power and all potential channel powers (including fading). Therefore, $\tail_m \in \mcalr_0$ implies that the channel powers take large values with non negligible probability. As a result, $m$ dominates and compensates the path loss, resulting in maximum power that grows at the same rate as the aggregate interference. This provides a theoretical justification to the fact that network densification always enhances the signal quality $Y$.

\begin{figure}[!t]
	\centering
	\subfigure[$\beta_0 = 3$, $F_m \sim \text{Composite}$]
	{
		\includegraphics[width=0.225\textwidth]{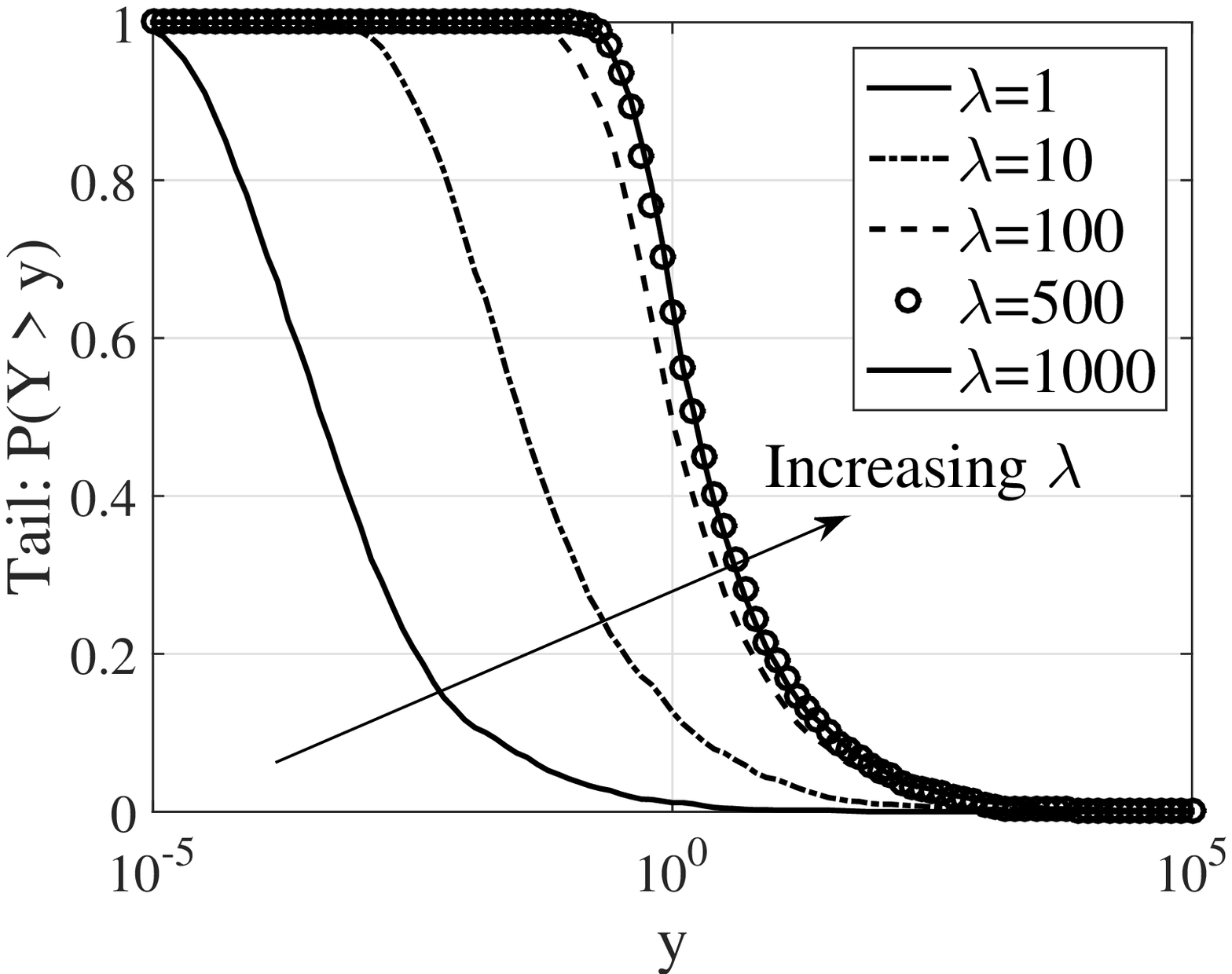}
	}
	\subfigure[$\beta_0 = 0$, $F_m \sim \text{Pareto}(0.5)$]
	{
		\includegraphics[width=0.225\textwidth]{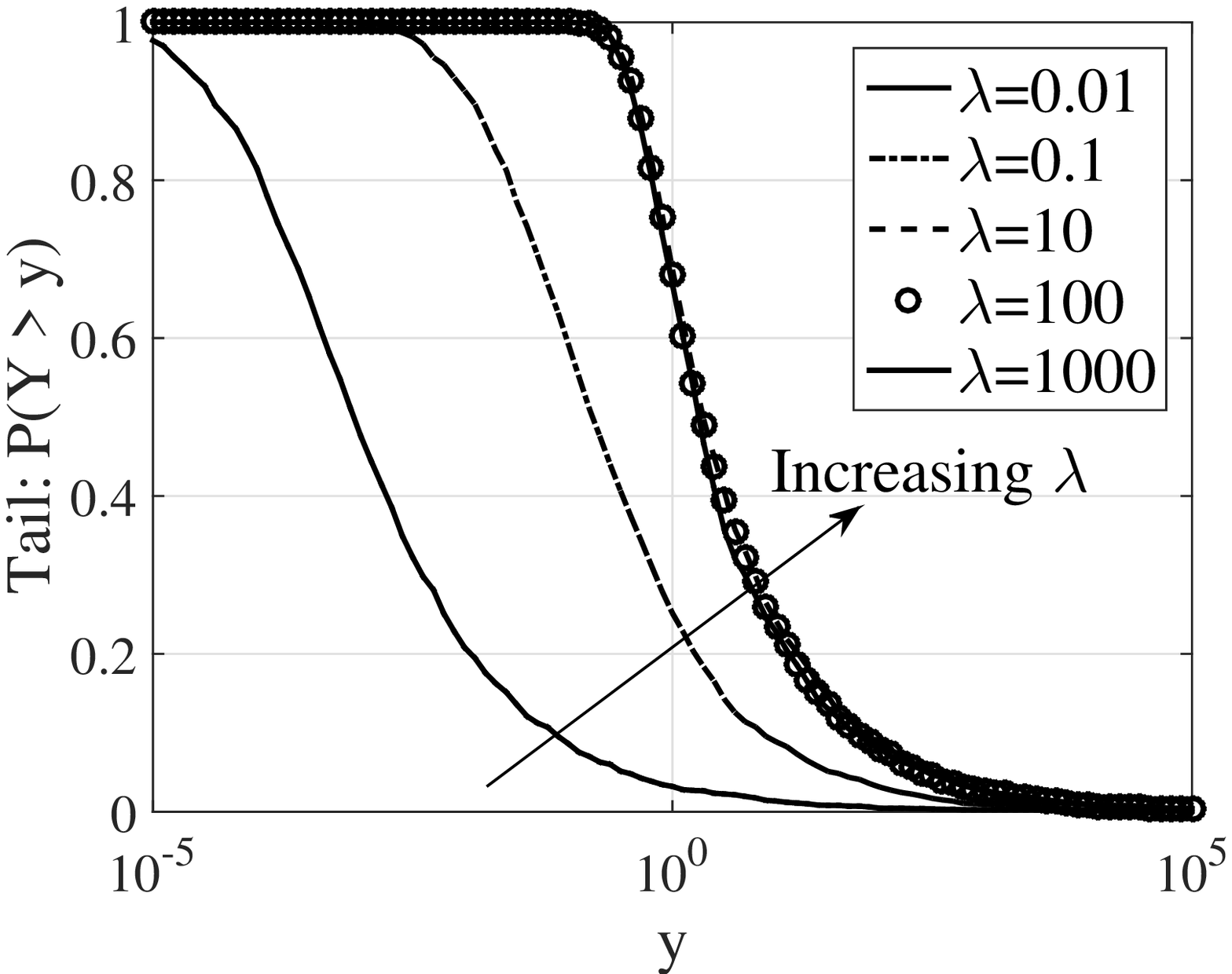}
	}
	\caption{SINR CCDF for $\tail_P \in \mcalr_{-\alpha}$, $\alpha \in (0,1)$. $K =2$, $A_0 = 1$, $\beta_1 = 4$, $R_1 = 10$ m, $R_{\infty} = 40$ km, $d=2$, $\lambda$ in $\text{base stations}/\text{km}^2$.}
	\label{fig:ConvergYtoD}
\end{figure}

When $\tail_P \in \mcalr_{-\alpha}$ with $0 < \alpha < 1$, $Y \overset{d}{\to} D$ implies that the SINR distribution converges to a non-degenerate distribution. Moreover, from \corref{cor:classifyTail}, this convergence is due to either large near-field exponent or heavy-tailed fading. In that case, for any SINR target $y$, the coverage probability $\Pb(Y > y)$ flattens out starting from some large value of node density. This means that further increasing the network density by deploying more BSs does not improve the network performance. In \fref{fig:ConvergYtoD} we show the convergence of $Y$ to a steady distribution for two cases: $\beta_0 > d$ or $\tail_m \in \mcalr_{-\alpha}$ with $\alpha \in (0,1)$, where `\emph{Composite}' represents the case of composite Rayleigh-lognormal fading, which belongs to the rapidly varying class $\mcalr_{-\infty}$, and $\text{Pareto}(\alpha)$ stands for Pareto fading distribution of shape $1/\alpha$ and some scale $\sigma > 0$, i.e.
\begin{equation}\label{eq:Pareto}
	\text{Pareto}(\alpha): \tail_m(x) = (1 + x/\sigma)^{-\alpha}.
\end{equation}

In practically relevant network settings, the path loss is bounded and fading is moderately heavy tailed (i.e. $\mcalr_{-\infty}$ as in the case of lognormal shadowing and Rayleigh fading) or even truncated (i.e. $\tail_m = o(\bar{H})$ where $H \in \mcalr_{-\infty}$). As a result, based on \thmref{thm:tailequiv}, we have that $\tail_P \in \mcalr_{-\infty}$, hence $Y \overset{a.s.}{\to} 0$. In other words, the SINR is proven to be asymptotically decreasing with the infrastructure density. This means that there is a fundamental limit on network densification and the network should not operate in the ultra-dense regime. 
In other words, deploying too many BSs would decrease the network performance due to the fact that the increased signal power cannot compensate for the faster growing aggregate interference. \fref{fig:ConvergYto0} confirms that with $\tail_P \in \mcalr_{-\alpha}$ with $\alpha > 1$ (i.e. either $\alpha_0 = 2/\beta_0 = 2$ or $\tail_m \sim \text{Pareto}(4)$), the tail of $Y$ vanishes and converges to zero when $\lambda$ increases (ultra-dense regime). 
This convergence to zero of the SINR in the ultra-dense regime further emphasizes the importance of local spatial scheduling among BSs since near-field interferers generate much stronger interference than far-field ones.

\begin{figure}[!t]
	\centering
	\subfigure[$\beta_0 = 1$, $F_m \sim \text{Composite}$]
	{
		\includegraphics[width=0.225\textwidth]{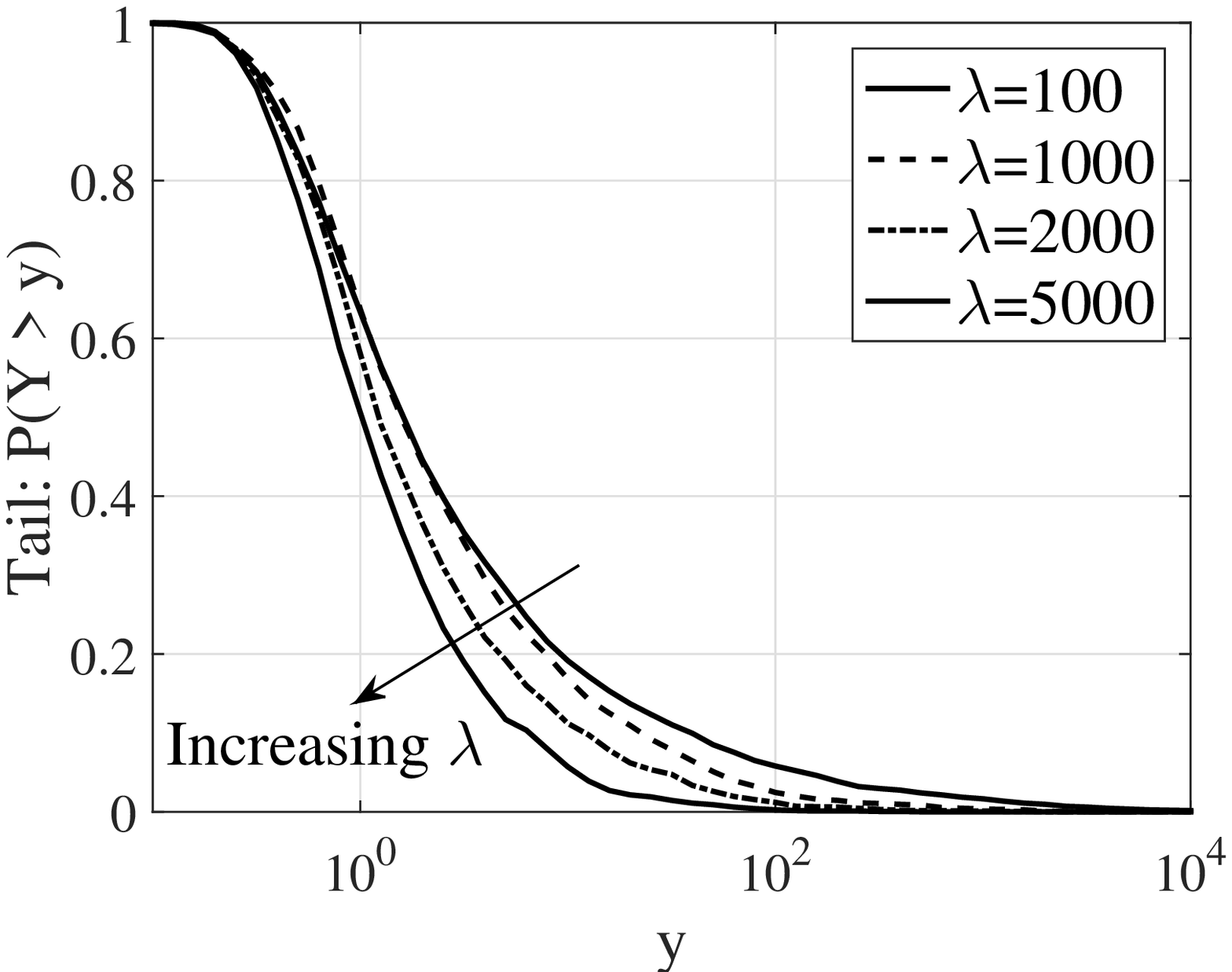}
	}
	\subfigure[$\beta_0 = 0$, $F_m \sim \text{Pareto}(4)$]
	{
		\includegraphics[width=0.225\textwidth]{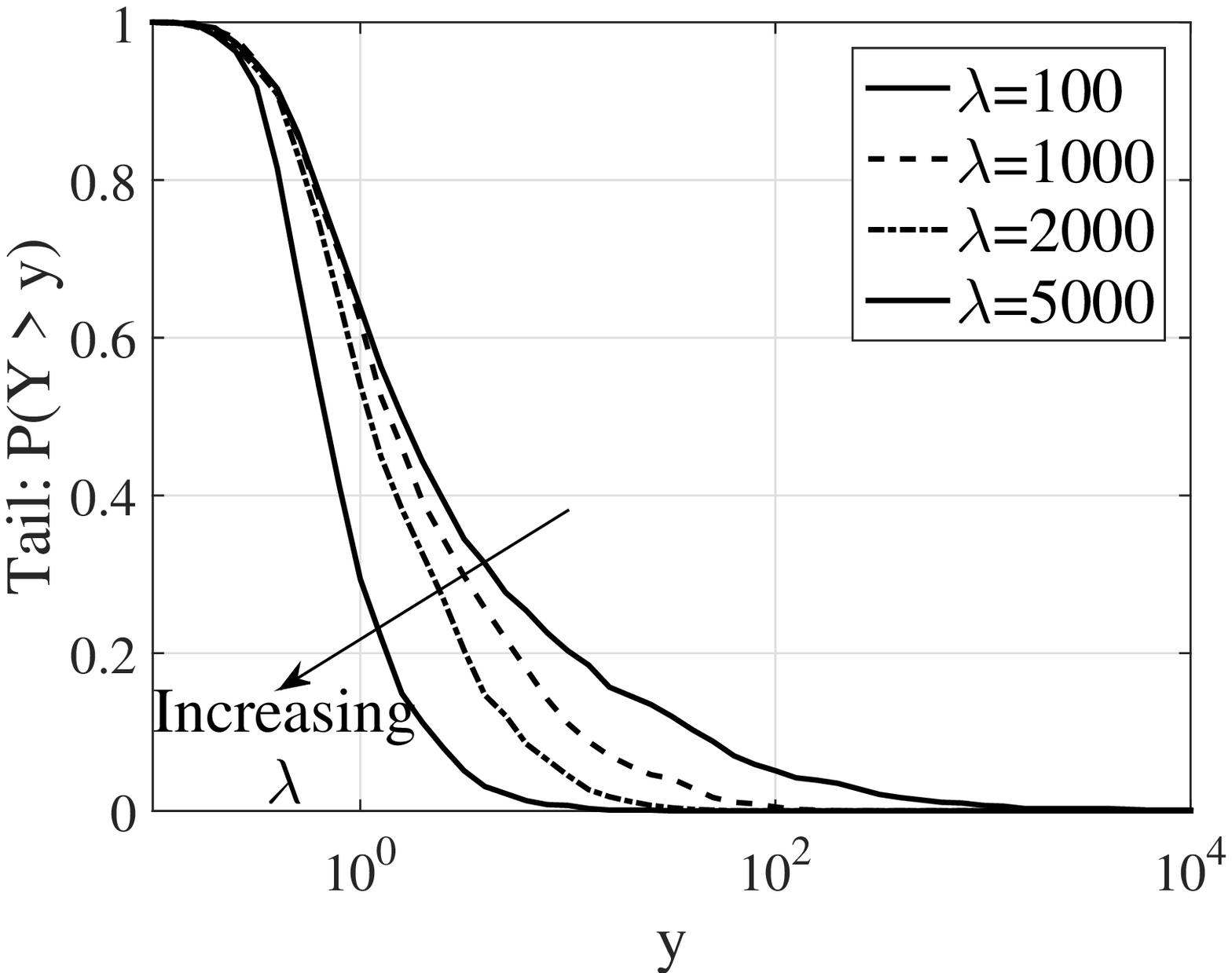}
	}
	\caption{SINR CCDF for $\tail_P \in \mcalr_{-\alpha}$ with $\alpha > 1$. $K =2$, $A_0 = 1$, $\beta_1 = 4$, $R_1= 10$ m, $R_{\infty} = 40$ km, $d=2$, $\lambda$ in $\text{base stations}/\text{km}^2$.}
	\label{fig:ConvergYto0}
\end{figure}

\subsection{Coverage and Capacity Scaling}
To provide a complete characterization of the network performance in the ultra-dense regime, we further study the scaling of coverage probability and capacity. By the definition of the aforementioned two metrics and \thmref{thm:YatLdaInfty}, we obtain the following
result.

\begin{cor}\label{cor:PerfLimit} Under the multi-slope path loss model and general fading, when $\lambda \to \infty$, for fixed $y > 0$, the coverage $\Pp_y(\lambda)$ and the capacity $\Cc(\lambda)$ scale as follows:
    \begin{enumerate}
    \item $\Pp_y(\lambda) \to 1$ and $\Cc(\lambda) \to \infty$ as $\lambda \to \infty$ if $\tail_P \in \mcalr_0$.
    \item $\Pp_y(u\lambda)/\Pp_y(\lambda) \to 1$ and $\Cc(u\lambda)/\Cc(\lambda) \to 1$ for $u > 0$ as $\lambda \to \infty$ if $\tail_P \in \mcalr_{-\alpha}$ with $0 < \alpha < 1$.
    \item $\Pp_y(\lambda) \to 0$ and $\Cc(\lambda) \to 0$ as $\lambda \to \infty$; moreover there exist finite densities $\lambda_p, \lambda_c$ such that $\Pp_y(\lambda_p) > \lim_{\lambda \to \infty}\Pp_y(\lambda)$ and $\Cc(\lambda_c) > \lim_{\lambda \to \infty}\Cc(\lambda)$ if $\tail_P \in \mcalr_{-\alpha}$ with $\alpha > 1$, or $\tail_P = o(\bar{H})$ where $\bar{H} \in \mcalr_{-\infty}$.
    \end{enumerate}
\end{cor}

\begin{proof}
    Case (1) directly follows \thmref{thm:YatLdaInfty}. 
    For case (2), by \thmref{thm:YatLdaInfty}, we have $Y \overset{d}{\to} D$, where $D$ has a non-degenerate distribution. Then, for constant $u > 0$ and $y > 0$, $\Pb(Y(u\lambda) \geq y) = \Pb(D \geq y) = \Pb(Y(\lambda) \geq y)$ as $\lambda \to \infty$. Similarly, $\E\pren{\log(1 + Y(u\lambda))} = \E\pren{\log(1+D)} = \E\pren{\log(1 + Y(\lambda))}$ as $\lambda \to \infty$.
		
	For case (3), given the conditions and \thmref{thm:YatLdaInfty}, we have $Y \overset{a.s}{\to} 0$. Hence, $\lim_{\lambda \to \infty}\Pp_y(\lambda) = 0$ and by \lemref{lem:YatLdaZero}, $\exists\lambda_p > 0$ s.t. $\Pp_y(\lambda_p) > \Pp_y(0) = 0$. For $\Cc(\lambda)$, we first note that $\Cc(\lambda) = \int_{0}^{\infty}\Pb(Y(\lambda) > y)/(1+y)\d{y}$. Thus, $\lim_{\lambda \to \infty}\Cc(\lambda) = 0$ by Lebesgue's dominated convergence theorem and by the fact that $Y \overset{p}{\to} 0$. For $\epsilon > 0$, $\exists\lambda_c > 0$ such that $\Pb(Y(\lambda_c) \geq \epsilon) > \Pb(Y(0) \geq \epsilon) = 0$ due to \lemref{lem:YatLdaZero}. Thus, $\forall y \in [0, \epsilon]$, $\Pb(Y(\lambda_c) \geq y) > 0$ since $\Pb(Y \geq y)$ is decreasing w.r.t. $y$. In consequence, $\Cc(\lambda_c) = \int_{0}^{\infty}\Pb(Y(\lambda_e) > y)/(1+y)\d{y} > 0$.
\end{proof}

The following observations can be made based on \corref{cor:PerfLimit}: First, both the coverage probability (cf. \fref{fig:ConvergYtoD}) and the capacity (cf. \fref{fig:SaturSE}) saturate (ceiling effect) when the fading is regularly varying with index in $(-1,0)$ or the near-field path loss exponent is large (i.e. $\beta_0 > d$). Second, when the fading is less heavy-tailed ($\tail_m \in \mcalr_{-\alpha}$ with $\alpha > 1$) and when the near-field path loss exponent is smaller (i.e. $\beta_0 < d$), both performance metrics are maximized at a finite network density, then decrease and go to zero in the ultra-dense regime (cf. `\emph{inverse U}' curves in \fref{fig:InvuSE} and \fref{fig:InvuCoverage}). This suggests that there is an optimal point of network density to aim for. Furthermore, fading and path loss have equivalent and somehow interchangeable effects on the network performance. Additionally, the most affecting element of the path loss is its near-field exponent $\beta_0$; bounded path loss function (i.e. $\beta_0 = 0$) is just a special case of the class $\beta_0 < d$. The impact of the near-field exponent was also observed in \cite{Zhang2015} for nearest cell association.

\tref{tab:PerfCases} summarizes the behavior of network performance according to \corref{cor:PerfLimit}.
{\renewcommand{\arraystretch}{1.1}
	\begin{table}[!t]
		\centering
		\caption{Scaling Regimes.}\label{tab:PerfCases}   
		\begin{tabular}{c||c|c|c|c}
			\hline\hline
			\multirow{2}{*}{\parbox[m]{1.2cm}{Scaling Regime}} & \multicolumn{3}{c|}{$\tail_m \in \mcalr_{-\alpha}$} & \multirow{2}{*}{\parbox[m]{0.9cm}{\centering lighter tail}} \\
			\cline{2-4}
			& \parbox[m]{1.2cm}{\centering $\alpha = 0$} & \parbox[m]{1.6cm}{\centering $0 < \alpha < 1$} & \parbox[m]{1cm}{\centering $\alpha > 1$} &  \\
			\hline\hline
			$\beta_0 < d$ & \multirow{2}{*}{\parbox[m]{1.2cm}{\centering $\Pp \to 1$, \\ $\Cc \to \infty$}} & \multirow{2}{*}{saturation} & \multicolumn{2}{c}{inverse U}  \\
			\cline{1-1}\cline{4-5}
			$\beta_0 > d$ &  &  & \multicolumn{2}{c}{saturation}\\
			\hline
		\end{tabular}
	\end{table}
}

\begin{figure}[!t]
	\centering
	\includegraphics[width= 0.42\textwidth]{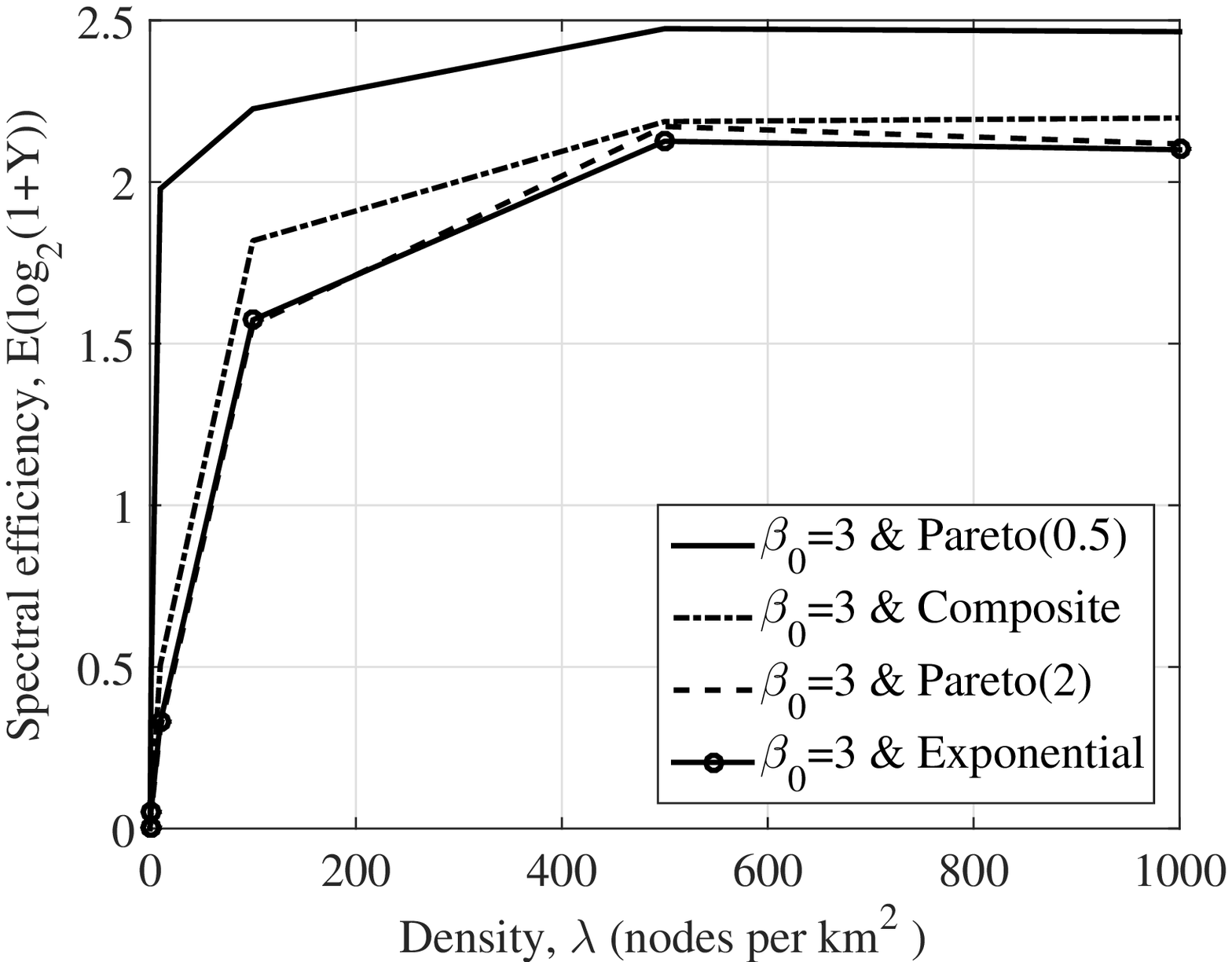}
	\caption{Capacity scaling for $\beta_0 > d$. Here, $d=2$, $K =2$, $A_0 = 1$, $\beta_1 = 4$, $R_1 = 10$m, $R_{\infty} = 40$ km.}
	\label{fig:SaturSE}
\end{figure}

\begin{figure}[!t]
	\centering
	\includegraphics[width= 0.42\textwidth]{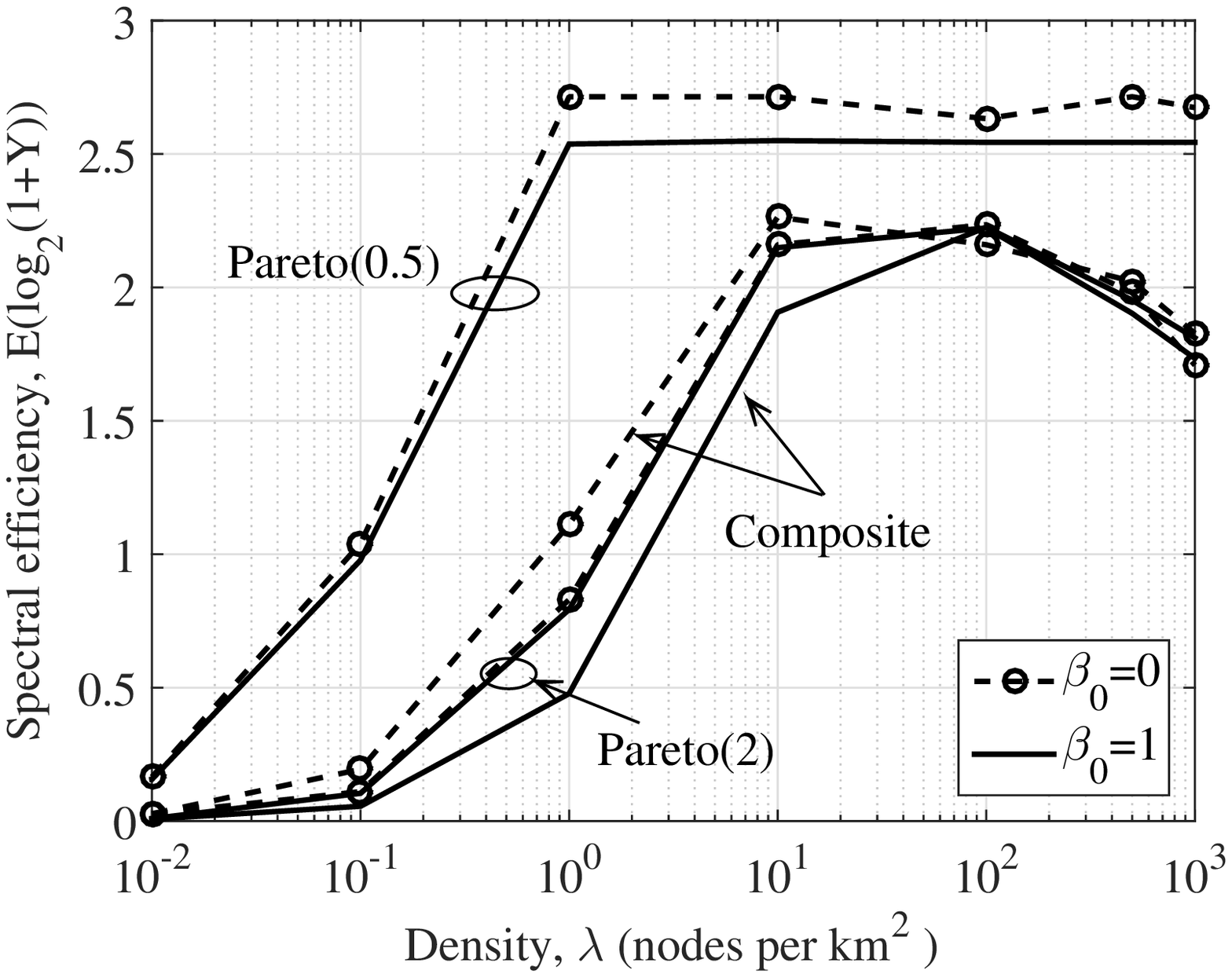}
	\caption{Capacity scaling for $\beta_0 < d$. Here, $d=2$, $K =2$, $A_0 = 1$, $\beta_1 = 4$, $R_1 = 10$ m, $R_{\infty} = 40$ km.}
	\label{fig:InvuSE}
\end{figure}

\begin{figure}[!t]
	\centering
	\includegraphics[width= 0.42\textwidth]{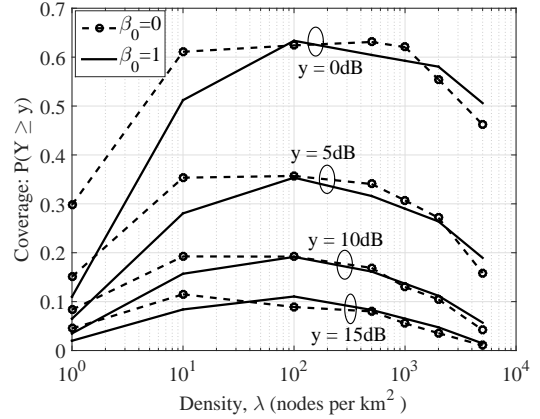}
	\caption{Coverage scaling for $\beta_0 < d$ and $F_m \sim \text{Composite}$. Here, $d=2$, $K =2$, $A_0 = 1$, $\beta_1 = 4$, $R_1 = 10$ m, $R_{\infty} = 40$ km.}
	\label{fig:InvuCoverage}
\end{figure}

\section{Conclusion}\label{s:Conclusion}
This paper analyzes the scaling regimes of the coverage probability and average rate under multi-slope path loss model and general fading distributions, considering strongest cell association in a Poisson field network. We show that when the fading distribution is regularly varying with index in $(-1, 0)$ or the near-field path loss exponent $\beta_0$ is larger than the network dimension $d$, both coverage and capacity saturate to a constant value in the ultra-dense regime. When the fading distribution is less heavy-tailed and when $\beta_0 < d$, both coverage and capacity are maximized at some finite density. Finally, our results show that path loss and fading have interchangeable effects on the tail behavior of SINR, coverage probability and capacity, extending previous results using dual slope path loss model and Rayleigh fading.

\appendices
\section{Proof of \propref{prop:tailP}}\label{appx:ProofTail}
    Using $l(r)$ given in \eref{eq:PL}, we have
    \begin{equation}\label{eq:tailP_expr}
        \tail_P(t) = \int_{0}^{R_{\infty}}\tail_m(tl(r))G(\d{r}) = \sum_{k=0}^{K-1} \mathcal{I}_k(t),
    \end{equation}
    where
    $\mathcal{I}_k(t) := \int_{R_k}^{R_{k+1}}\tail_m\pren{A_k r^{\beta_k} t} G(\d{r})$.
    
    (i) For $\beta_k > 0$: integration by parts with $\tail_m(A_k r^{\beta_k} t)$ and $G(\d{r})$ for $\mathcal{I}_k$ yields
    \begin{equation*}
        \mathcal{I}_k(t) = \tail_m\pren{A_k r^{\beta_k} t} G(r) \Big|_{R_k}^{R_{k+1}} + J_k(t),
    \end{equation*}
    where
    $J_k(t) = \int_{R_k}^{R_{k+1}}G(r)\d{F_m(A_k r^{\beta_k} t})$, 
    which reduces to \eref{eq:Jk} after change of variable $u = A_k r^{\beta_k} t$ and applying the condition $A_k R_{k+1}^{\beta_k} = A_{k+1} R_{k+1}^{\beta_{k+1}} = a_{k+1}$ due to \eref{eq:plScale}. On the other hand, applying again \eref{eq:plScale},
    \begin{equation}\label{eq:tailP_sum1}
        \sum_{k=0}^{K-1}\pren{\tail_m\pren{A_k r^{\beta_k} t} G(r) \Big|_{R_k}^{R_{k+1}}} 
        = \tail_m(a_K t).
    \end{equation}
    Hence, substituting \eref{eq:tailP_sum1} back to \eref{eq:tailP_expr} yields \eref{eq:tailP} with $k_0=0$.

    (ii) For $\beta_0 = 0$, we have:
    \begin{equation*}
	    \tail_P(t) = \mathcal{I}_0(t) + \sum_{k=1}^{K-1}\mathcal{I}_k(t) = 
	    \int_{0}^{R_1}\tail_m\pren{A_0 t} G(\d{r}) + \sum_{k=1}^{K-1}\mathcal{I}_k(t).
    \end{equation*}
    Thus,
    \begin{multline*}
        \tail_P(t) = \tail_m(A_0 t) G(R_1) + \tail_m\pren{A_{K-1} R_{\infty}^{\beta_{K-1}} t}G(R_{\infty}) \\
        - \tail_m\pren{A_{1} R_{1}^{\beta_1} t}G(R_1) + \sum_{k=1}^{K-1}J_k(t),
    \end{multline*}
    in which $A_0 = A_{1} R_{1}^{\beta_1}$ due to \eref{eq:plScale}. Hence, \eref{eq:tailP} with $k_0=1$.

\section{Proof of \thmref{thm:tailequiv}}\label{appx:ProofTailEquiv}
    Due to space limitations, we only give the proof for the case of $\tail_m \in \mcalr_{-\alpha}$ for $\alpha \in [0,\infty)$. The proof for the cases $\tail_m \in \mcalr_{-\infty}$ and $\tail_m = o(\bar{H}), H \in \mcalr_{-\infty}$ are similar and would be provided in a longer version of this paper. 
    
    $\tail_m$ can be represented as $\tail_m(x) \sim x^{-\alpha} L(x)$ for some $L \in \mcalr_0$. Then, by the monotone density theorem \cite{Bingham1989,Embrechts1997}, the density function $f_m$ of $F_m$ can be written as
    $f_m(t) \sim \alpha t^{-\alpha-1}L(t)$ as $t \to \infty$.
    
    \emph{For $k \geq 1$}, as $t \to \infty$,
    \begin{equation*}
        \E\pren{m^{\alpha_k}\one(a_k t \leq m < a_{k+1} t)} \sim \int_{a_k t}^{a_{k+1}t}\alpha x^{\alpha_k - \alpha - 1} L(x) \d{x}.
    \end{equation*}
    If $\alpha < \alpha_k$: by Karamata's theorem \cite[Prop.~1.5.8]{Bingham1989}, for $t_0 > 0$:
    \begin{align*}
        &\int_{a_k t}^{a_{k+1}t}\alpha x^{\alpha_k - \alpha - 1} L(x) \d{x} \\
        &= \int_{t_0}^{a_{k+1}t}\alpha x^{\alpha_k - \alpha - 1} L(x) \d{x} - \int_{t_0}^{a_{k}t}\alpha x^{\alpha_k - \alpha - 1} L(x) \d{x} \nonumber \\
        &\sim \frac{\alpha}{\alpha_k - \alpha}\pren{a_{k+1}^{\alpha_k - \alpha} - a_k^{\alpha_k - \alpha}} L(t) t^{\alpha_k - \alpha}, \quad \text{ as } t \to \infty.
    \end{align*}
    If $\alpha > \alpha_k$: similarly to the above case, we can easily obtain the same result using Karamata's theorem. Note that if $\alpha = \alpha_{\hat{k}}$ for some $\hat{k} \in [1,K-1]$, we can also show that $J_{\hat{k}}(t) \sim c L(t) t^{-\alpha}$ for some constant $c$. Thus, for $k \geq 1$,
    \begin{align}
        J_k(t) 
        \sim \frac{\alpha}{\alpha_k - \alpha} \frac{a_{k+1}^{\alpha_k - \alpha} - a_k^{\alpha_k - \alpha}}{A_k^{\alpha_k}R_{\infty}^d} L(t) t^{-\alpha}, \text{ as } t \to \infty. \label{eq:proofTailEquivk1}
    \end{align}
    \emph{For $k=0$}: If $\alpha_0 < \alpha$, then $\E(m^{\alpha_0})$ exists and $\E\pren{m^{\alpha_0}\one(0 \leq m < a_{1} t)} = \E(m^{\alpha_0})$ as $t \to \infty$. Thus, $J_0(t) = \E(m^{\alpha_0})A_0^{-\alpha_0} R_{\infty}^{-d} t^{-\alpha_0}$ as $t \to \infty$. Otherwise, i.e. if $\alpha_0 \geq \alpha$, then we have
    \begin{multline*}
        \E\pren{m^{\alpha_0}\one(0 \leq m < a_{1} t)} = \int_{0}^{a_1t}x^{\alpha_0}F_m(\d{x}) = \\
        \overset{(a)}{=} \frac{\alpha (a_1 t)^{\alpha_0}}{\alpha_0 - \alpha} \tail_m(a_1t) \overset{(b)}{=} \frac{\alpha (a_1 t)^{\alpha_0 - \alpha}}{\alpha_0 - \alpha} L(t),
    \end{multline*}
    where $(a)$ is due to \cite[Prop.~A3.8]{Embrechts1997}, and $(b)$ is due to the representation $\tail_m(x) \sim x^{-\alpha} L(x)$. Thus, as $t \to \infty$,
    \begin{align}
        J_0(t) 
        = \begin{cases} \E(m^{\alpha_0}) A_0^{-\alpha_0}R_{\infty}^{-d} t^{-\alpha_0}, & \text{ if } \alpha_0 < \alpha \\ \alpha (\alpha_0 - \alpha)^{-1} a_1^{\alpha_0-\alpha} L(t) t^{-\alpha}, & \text{ if } \alpha_0 \geq \alpha \end{cases}.\label{eq:proofTailEquivk0}
    \end{align}    
    By substituting \eref{eq:proofTailEquivk1} and \eref{eq:proofTailEquivk0} into the expressions of $\tail_P$ given by \propref{prop:tailP}, we have as $t \to \infty$,
    \begin{align*}
        &\tail_P(t) = \tail_m\pren{a_K t} + \sum_{k=i}^{K-1}J_k(t) \\
        & \sim \begin{cases} \pren{a_K^{-\alpha} + \sum_{k=1}^{K-1} C_k } \frac{L(t)}{t^{\alpha}}, \quad \text{if } \beta_0 = 0,\\
        \pren{a_K^{-\alpha} + \sum_{k=0}^{K-1} C_k} \frac{L(t)}{t^{\alpha}}, \quad \text{if } \beta_0 > 0, \alpha_0 \geq \alpha \\
        \pren{a_K^{-\alpha} + \sum_{k=1}^{K-1} C_k} \frac{L(t)}{t^{\alpha}} + \frac{\E(m^{\alpha_0})}{A_0^{\alpha_0} R_{\infty}^d t^{\alpha_0}}, \text{if } \beta_0 > 0, \alpha_0 < \alpha,
        \end{cases}
    \end{align*}
    where $C_k = \alpha(a_{k+1}^{\alpha_k - \alpha} - a_k^{\alpha_k - \alpha})\pren{(\alpha_k - \alpha)A_k^{\alpha_k}R_{\infty}^d}^{-1}$,
    and where for the last case with $\beta_0 > 0$ and $\alpha_0 < \alpha$, we further have $t^{-\alpha} = o(t^{-\alpha_0})$ as $t \to \infty$. As a result, $\tail_P$ is regularly varying with index $\alpha$ if $\alpha \leq \alpha_0$, and with index $\alpha_0$ if $\alpha_0 < \alpha$. Hence the proof for $\tail_m \in \mcalr_{-\alpha}$ with $\alpha \in [0,\infty)$.
\bibliographystyle{IEEEtrans}
\bibliography{IEEEabrv,Densification}

\end{document}